\documentclass[a4paper,11pt]{article}
\usepackage{graphicx}
\usepackage{here}
\usepackage{amsmath,amssymb,amsthm}
\usepackage{mathrsfs}
\usepackage{ascmac}
\usepackage[sort, comma]{natbib}
\usepackage[subrefformat=parens]{subcaption}
\usepackage{url}
\usepackage[stable]{footmisc}
\usepackage{mathtools}
\mathtoolsset{showonlyrefs}

\newtheorem{theo}{Theorem}

\newtheorem{prop}{Proposition}
\newcommand{\ol}[1]{\overline{#1}}
\newcommand{\ul}[1]{\underline{#1}}
\newcommand{\cl}[1]{\mathcal{#1}}

\title{{\LARGE A micro-founded comparison of fiscal policies between indirect and direct  job creation}}
\author{Kensuke Ohtake\thanks{Center for General Education, Shinshu University, Matsumoto, Nagano 390-8621, Japan,
E-mail: k\_ohtake@shinshu-u.ac.jp}}
\date{February 28, 2024}

\begin{document}
\maketitle

\begin{abstract}
The purpose of this paper is to provide a micro-economic foundation for an argument that the direct employment by the government is more desirable than the government purchase of private goods to eliminate unemployment. A general equilibrium model with monopolistic competition is devised, and the effects of policies (government purchase, tax rate operation, and government employment) on macroeconomic variables (consumption, price, and profit) are investigated. It is shown that 1) the government purchase is inflationary in the sense that additional effective demand by the government not only increases private employment but also raises prices; 2) the government employment can achieve full employment without causing a rise in prices. 
\end{abstract} 



\section{Introduction}
We consider two different forms of fiscal policy. One is a policy in which the government purchases private goods with the intention of securing corporate profits and increasing employment in the private sector by expanding production. The other is a policy in which the government directly employs workers in order to eliminate unemployment. Let us refer to the former as ``government purchase'' and the latter as ``government employment'' for simplicity. These two policies are often blurred by being treated just as ``government expenditure'', and not distinguished even in policy debates, let alone in textbook economics. Actually, much literature seems to imply the government purchase by the government expenditure.\footnote{We must mention several works in the literature that distinguish between the government purchase and government employment. For example, \citet{pall18govspe} constructs Keynesian and Kaleckian models, discusses multiplier effects and distributional aspects, and addresses political economic concerns related to Job Guarantee Program (JGP). Although away from the context of the present paper, several DSGE models that incorporate the government employment have been developed, such as \citet{cav05gov}, \citet{eguhira2009}, \citet{ForMonSe09}, \citet{StTh12Fimod}, \citet{BuyHey2014}, \citet{ChaLiTraYa21}, \citet{Wang21}, and \citet{GerWa22}.}

Are differences between these two fiscal policies insignificant? \citet{Minsky1968} and some Minskyan authors have argued that the differences should not be ignored.\footnote{For an overview of Minsky's arguments on employment policy, see, for example, \citet{bewr2004} and \citet{wray2015minskybook}.}  Their arguments might be summarized as follows: 1) Pump-priming government purchase intended to stimulate the economy is insufficient to achieve full employment, and instead, it destabilizes the economy through inflation, rising inequality, and weakened financial structure. 2) On the other hand, government employment\footnote{\citet{Minsky1968} specifically calls the government employment policy ``Employer of Last Resort (ELR)''. Recent Minskyan authors, develop the concept of ELR and call it ``Job Guarantee Program (JGP)'' as in \citet[Chapter 19]{MitWraWatt2019book}. These concepts have carefully designed contents in Minskyan thought, but this paper does not discuss detailed contents of the government employment.} policies can defeat these economic destabilizing factors as well as achieve full employment directly.

As a mathematical expression of such Minskyan arguments, \citet{Tcher2012} (\citet{Tcher2014}) provides an enlightening model regarding the comparison of several types of fiscal policies: unemployment insurance, investment subsidy, government purchase, and government employment, the first three and the fourth of which are referred to as ``indirect job creation'', and ``direct job creation'', respectively. \citet[p.142]{Tcher2014} concludes that the latter is desirable mainly in  that it contributes to stabilizing the price-cost markup, stating that
\begin{quote}
{\it With ELR,\footnote{Note by the author: Employer of Last Resort. See also Footnote 2.} full employment is guaranteed, whereby the inflationary influences on the markup from government are controlled.}
\end{quote}
and \citet[p.139]{Tcher2014} also states that
\begin{quote}
{\it There are, of course, other government and non-government factors that can
produce inflation, but inflation it must be emphasized does not result from the ELR program.}
\end{quote}
Tcherneva's model is a clever application of macroeconomic accounting identities. In addition to such macro economic discussions, the author of the present paper believes that understanding how the macro economic result can be explained from micro economic behavior of firms and consumers may help to further develop the Minskyan arguments.

In this paper, we simplify one of the Minskyan claims on the desirability of the government employment as {\it the government employment is less inflationary than the government purchase},\footnote{\citet{Tcher2014} does not necessarily make such a simple assertion. In Tcherneva's model, the private economy consists of two sectors, i.e., consumption goods and investment goods, so the impact of government purchase on prices differs for each sector. \citet[Section 19.6]{MitWraWatt2019book} are more direct in their arguments with the conclusion that ``JG eliminates the Phillips curve''. Here, JG stands for Job Guarantee.} and give an analytical proof of this claim in a micro-founded general equilibrium model. An intuitive explanation of the results of this paper would be as follows. The government purchase creates a certain amount of effective demand for firms regardless of their pricing. This allows firms to price aggressively. On the other hand, effective demand from the workers in the government employment forces firms to compete on price, which generates downward price pressure.

Our model is a static general equilibrium model with monopolistic competition\footnote{The model introduced by \citet{BlaKiyo87} is a standard static macroeconomic model based on monopolistic competition.} based on one by \citet{BehMura2007}, but extended in a non-Walrasian way. In our model, the utility function of each consumer is the constant absolute risk aversion (CARA) type, and therefore price-cost markup of each firm is not constant, unlike a model with the constant elasticity of substitution (CES) utility function.\footnote{\citet{BehMura2007} investigate differences between the constant relative risk aversion utility (including CES) model and CARA utility model} The purpose of adopting the CARA model is to develop discussion in a more general situation where price-cost markup is not constant because one of the main subjects of this paper is particularly about the impacts of policies on price.\footnote{\citet[Section 7.4]{Ike06} have also introduced a non-CES utility function, leading to a model with a non-constant price-cost markup.}

The term ``non-Walrasian'' is used in a somewhat simplified sense in this paper.\footnote{The standard references on non-Walrasian macroeconomics are, for example, \citet{Clo65}, \citet{Leij68}, \citet{BarGro71}, \citet{BarGro76}, and \citet{Bena86}.} There, the focus is not on any market disequilibrium in general, but on cases of insufficient demand for labor.\footnote{A more general description of the non-Walrasian approach is given in \citet[Chapter 2]{Bena86}. Similar to the present paper, \citet[pp.52-59]{Tana10} proposes an AD-AS model with the simplified non-Walrasian approach that focuses on the case of labor oversupply. Unlike the model of the present paper, in the one by \citet[pp.52-59]{Tana10} the equilibrium output increases as the nominal wage falls.} Trades take place under the current nominal wage level without waiting for the labor market reaches full employment. This might lead one to believe that wages are fixed even in the long run as in the textbook IS-LM model. On the contrary, the model in this paper allows wages to change in the long run. We use the non-Walrasian approach to express simply the fact that the real labor (and goods) markets do not operate like a well-organized stock market. In this way, the demand for goods from privately employed workers plus that from governmentally employed workers and that from the government constitute effective demand, and firms price optimally under the constraints of that effective demand.

The rest of the paper is organized as follows. Section 2 introduces the model that we handle. Section 3 derives the equilibrium of the model and discusses how it depends on the policies. Section 4 provides conclusions and discussions. Section 5 deals with complicated proofs omitted in the main text.

\section{The model}\label{sec:model}
This section introduces the model. First, the settings and assumptions of the model are specified, and then consumer and producer behavior are modeled.

\subsection{Settings and assumptions}
The economy consists of two sectors; one is private and the other is governmental. There are so many differentiated varieties of firms and goods in the private sector that we consider them to be modeled as a continuum. Each firm is monopolistically competitive and it produces only one variety of differentiated goods, hence the number of the firms is equal to the number of the varieties. Each firm (and variety) is indexed by $i\in [0, N]$ where $N\in \mathbb{R}$ is the size of the set of the firms (and varieties). Here, the size $N$ is assumed to be exogenous. We use the term {\it goods} exclusively in the sense of the privately produced goods because we assume that governmentally produced goods (or services) are not traded in the market and are not objects of consumer's choice.

Workers are also divided into two categories: privately employed workers and governmentally employed ones. All the workers earn the same nominal wage $w\geq 0$ regardless of whether they are in the private or governmental employment. Taxes are also levied evenly on all workers, so their after-tax incomes are also equal.

Consumers consist only of the workers, therefore,  positive profits are retained as financial assets of the firms and are not given to any expenditure. In addition, all the consumers have the same preference relation. Therefore, since their incomes are also equal, the consumption of each variety is also equal for all consumers.

In equilibrium, it is assumed that unemployment occurs in the labor market. On the other hand, the goods market always balances in the sense that the firms produce only as much as the effective demand given to them.

\subsection{Consumer behavior}
Let $u$ be a partial utility function for any variety $i\in[0, N]$ of the goods. For a given price $p(i)\geq 0$ of variety $i$, income $y\geq 0$, and tax $t\geq 0$, each consumer maximizes the following utility function $U$
\begin{equation}\label{utility}
U = \int_0^N u(q(i))di,
\end{equation}
under the budget constraint
\begin{equation}\label{budget}
\int_0^N p(i)q(i)di + t = y.
\end{equation}
The Lagrangian of the maximization problem for \eqref{utility} under \eqref{budget} is given by
\[
\cl{L} = \int_0^N u(q(i))di -\lambda\left[\int_0^N p(i)q(i)di + t - y\right],
\]
where $\lambda$ is the Lagrange multiplier. The first-order condition for the optimality is then given by
\begin{equation}\label{fo}
\frac{u^\prime(q(i))}{p(i)} = \frac{u^\prime(q(j))}{p(j)} = \lambda,~~~\forall i,j \in [0,N],
\end{equation}
From \eqref{fo}, we see that
\begin{equation}\label{uprimeqiuprimepipj}
u^\prime(q(i)) = u^\prime(q(j))\frac{p(i)}{p(j)}.
\end{equation}
In the following, we consider the partial utility function $u$ of the CARA type defined by
\[
u(q) := k-\kappa e^{-\alpha q},
\]
where $k\in \mathbb{R}$, $\kappa>0$, and $\alpha>0$. In this case, as in \citet[pp.780-781]{BehMura2007}, the equation \eqref{uprimeqiuprimepipj} can be solved for $q(i)$ as 
\begin{equation}\label{qiqi1alppipj}
q(i) = q(j) -\frac{1}{\alpha}\ln \left[\frac{p(i)}{p(j)}\right].
\end{equation}
By multiplying both sides of the equation \eqref{qiqi1alppipj} by $p(i)$ and integrating over $[0,N]$, from $\int_0^N p(i)q(i)di=y-t$, we get 
\[
y-t = q(j)\int_0^N p(i)di - \frac{1}{\alpha}\int_0^N \ln\left[\frac{p(i)}{p(j)}\right]p(i)di,
\]
which immediately yields the demand function\footnote{This corresponds to \citet[p.782, (13)]{BehMura2007}.}
\begin{equation}\label{dfunc}
q(j) = \frac{y-t+\frac{1}{\alpha}\int_0^N \ln\left[\frac{p(i)}{p(j)}\right]p(i)di}{\int_0^N p(i)di}.
\end{equation}

\subsection{Producer behavior}
Let $m>0$ and $F\geq 0$ be the marginal input and fixed input of workers, respectively, to produce one unit of goods. That is, when each firm produces $q\geq 0$ of a variety, the required labor input $l$ is given by
\[
l = F + mq.
\]
Each worker earns an arbitrarily given nominal wage $w\geq 0$. For the number of  
privately employed workers $L\geq 0$, the number of governmentally employed workers $L_g\geq 0$, and the quantity of government purchase $g(i)\geq 0$ for variety $i$, the profit of firm $i$ is given by
\begin{equation}\label{defprofit}
\Pi(i) = \left(Lq(i)+L_gq(i)+g(i)\right)\left[p(i)-mw\right] - Fw.
\end{equation}
By differentiating $q(i)$ by $p(i)$, we see from \eqref{dfunc} that
\begin{equation}\label{dqidpi}
\begin{aligned}
\frac{dq(i)}{dp(i)} &= -\frac{1}{\alpha p(i)}.
\end{aligned}
\end{equation}
From \eqref{dfunc}, \eqref{defprofit}, and \eqref{dqidpi}, we obtain the first order condition $\frac{d\Pi(i)}{dp(i)}=0$ as
\begin{equation}\label{fofirm}
\begin{aligned}
&\int_0^N \ln\left[\frac{p(i)}{p(j)}\right]p(j)dj + \frac{p(i)-mw}{p(i)}\int_0^N p(j)dj \\
&\hspace{33mm} = \alpha(y-t) + \frac{\alpha g(i)}{L+L_g} \int_0^N p(j)dj.
\end{aligned}
\end{equation}
for the maximal profit.\footnote{This corresponds to \citet[p.783, (17)]{BehMura2007}.}

\section{Equilibrium analysis}
This section discusses basic properties of an equilibrium called the symmetric equilibrium and then investigates how the equilibrium responds to each of three policy variables: the government purchase, tax rate, and government employment.

\subsection{Basic properties of equilibrium}
We consider an equilibrium often referred to as {\it symmetric equilibrium} in general equilibrium models.\footnote{For example \citet{BlaKiyo87}, \citet{BehMura2007}, \citet{Ike06}, or \citet{Tana10}.} That is, $p(i)\equiv p$ and $q(i)\equiv q$ for all $i\in [0,N]$. The following proposition provides justification for considering only the symmetric equilibrium in the present model.\footnote{The proposition corresponds to \citet[Proposition 2]{BehMura2007}, and the proof is also given in the same manner.}
\begin{prop}\label{jussym}
Assume $g(i)\equiv g,~\forall i\in[0,N]$, then $p(i)\equiv p,~\forall i\in[0, N]$.
\end{prop}
\begin{proof}
Let us define $\overline{p}:=\max_i p(i)$ and $\underline{p}:=\min_i p(i)$. The first order condition \eqref{fofirm} is valid for all the firms in common, therefore,
\[
\begin{aligned}
&\int_0^N \ln\left[\frac{\ol{p}}{p(j)}\right]p(j)dj + \frac{\ol{p}-mw}{\ol{p}}\int_0^N p(j)dj =  \alpha(y-t) + \frac{\alpha g}{L+L_g} \int_0^N p(j)dj\\
&\int_0^N \ln\left[\frac{\ul{p}}{p(j)}\right]p(j)dj + \frac{\ul{p}-mw}{\ul{p}}\int_0^N p(j)dj = \alpha(y-t) + \frac{\alpha g}{L+L_g} \int_0^N p(j)dj
\end{aligned}
\]
hold. Subtracting the lower equation from the upper one, we obtain
\begin{equation}
\left[\ln\left(\frac{\ol{p}}{\ul{p}}\right)+\frac{mw(\ol{p}-\ul{p})}{\ol{p}\ul{p}}\right]\int_0^N p(j)dj=0,
\end{equation}
which is impossible if $\ol{p}>\ul{p}$. Therefore, $\ol{p}=\ul{p}$ ($\equiv p$).
\end{proof}
In the following, we assume that the after-tax income $y-t$ is the same for all consumers and $g(i)\equiv g$ for any $i\in[0, N]$. Then, we have the symmetric equilibrium because $q(i)\equiv q$ for any $i\in[0, N]$ by Proposition \ref{jussym} and the demand function \eqref{dfunc}.

It follows from \eqref{fofirm} that the equilibrium price $p$ satisfies
\begin{equation}\label{priceeqpmw}
\frac{p-mw}{p}Np = \alpha(y-t) + \frac{\alpha Npg}{L+L_g}.
\end{equation}
Solving \eqref{priceeqpmw} for $p$ yields
\begin{align}
p &= \frac{L+L_g}{L+L_g-\alpha g}\left[mw+\frac{\alpha(y-t)}{N}\right]. \label{p1}
\end{align}
Furthermore, substituting $y-t=\int_0^Np(i)q(i)di=Npq$ into \eqref{p1}, and solving for $p$, we obtain another expression for the equilibrium price as
\begin{align}
p &= \frac{mw\left(L+L_g\right)}{L+L_g-\alpha\left(\left(L+L_g\right)q+g\right)}.  \label{p2}
\end{align}

The following proposition states that a large size of the set of the firms and the varieties of the goods or a large marginal input guarantees a positive equilibrium price.
\begin{prop}\label{prop:LalpLqgplus}
If 
\begin{equation}\label{Nmbtalp}
Nm>\alpha,
\end{equation}
then 
\[
L+L_g - \alpha\left(\left(L+L_g\right)q+g\right) > 0
\]
holds,\footnote{If $\left(L+L_g\right)-\alpha(\left(L+L_g\right)q+g)>0$ in \eqref{p2} then of course $\left(L+L_g\right)-\alpha g>0$ in \eqref{p1}.} and thus the price \eqref{p2} {\rm(}and \eqref{p1}{\rm)} is positive.
\end{prop}
\begin{proof}
The total employment by the private firms is given by
\begin{equation*}
\int_0^N\left[m\left(L+L_g\right)q(i)+mg(i)+F\right]di,
\end{equation*}
which equals $L$. Therefore, in the symmetric equilibrium, we obtain
\begin{equation}\label{empeq}
L = N\left[m\left(L+L_g\right)q+mg+F\right].
\end{equation}
We see from \eqref{empeq} that
\begin{equation}\label{LNmLLgqgNF>0}
L-Nm(\left(L+L_g\right)q+g)=NF>0.
\end{equation}
Therefore, it follows from \eqref{Nmbtalp} and \eqref{LNmLLgqgNF>0} that
\[
\begin{aligned}
&L+L_g-\alpha\left(\left(L+L_g\right)q+g\right) \\
&\hspace{10mm} > L+L_g-Nm\left(\left(L+L_g\right)q+g\right) = NF+L_g >0.
\end{aligned}
\]
\end{proof}

The next proposition states that the equilibrium price is higher than the marginal cost.
\begin{prop}
Under \eqref{Nmbtalp}, 
\begin{equation}\label{p-mw>0}
p-mw > 0
\end{equation}
holds.
\end{prop}
\begin{proof}
From \eqref{p2} and Proposition \ref{prop:LalpLqgplus}, we see that
\begin{equation}
p-mw = \frac{\alpha mw\left((L+L_g)q+g\right)}{L+L_g-\alpha\left(\left(L+L_g\right)q+g\right)} > 0. \label{pmwresult}
\end{equation}
\end{proof}

We incorporate the proportional income tax $t=\tau w,~\tau\in(0,1)$, hence the after-tax income of each consumer is
\begin{equation}\label{uniformincome}
y-t=(1-\tau) w.
\end{equation}
Therefore, the consumption \eqref{dfunc} becomes
\begin{align}
q = \frac{(1-\tau)w}{Np}.\label{eqqp}
\end{align}
in the symmetric equilibrium. Substituting the price \eqref{p2} into \eqref{eqqp} and solving for $q$, we have
\begin{equation}\label{eqq}
q = \frac{(1-\tau)\left(L+L_g-\alpha g\right)}{\left(Nm+\alpha(1-\tau)\right)\left(L+L_g\right)}.
\end{equation}
Solving \eqref{empeq} for $L$ gives
\begin{equation}\label{L}
L = \frac{N\left(mL_gq+mg+F\right)}{1-Nmq}.
\end{equation}
From \eqref{L}, we can easily see that $L>0$ because it is immediate that
\begin{equation}\label{1mNmqgt0}
1-Nmq>0
\end{equation}
from \eqref{eqq}. Combining \eqref{eqq} and \eqref{L}, and solving the quadratic equation thus obtained for $L$, we obtain 
\begin{equation}\label{solvedL}
L = \frac{N\left[(1-\tau)(mL_g+\alpha F)+(mg+F)mN\right]}{\alpha+(Nm-\alpha)\tau}.
\end{equation}

\noindent
{\bf Remark}: For an exogenously given worker population $L^*$, there is no guarantee that the amount of employment $L$ given in \eqref{solvedL} matches it. In particular, we focus on the case $L<L^*$. Moreover, not only can there be such underemployment in the short-run where the nominal wage is arbitrarily fixed, but it cannot be eliminated even in the long-run where the nominal wage falls because the equilibrium employment \eqref{solvedL} does not depend on $w$, hence $dL/dw=0$. This would reflect the central Keynesian logic that a fall in wages results in a contraction of effective demand through a decrease in incomes.

\vspace{5mm}
The next proposition states that for the profit to be non-negative, the fixed input must be somewhat small.
\begin{prop}\label{prop:nonnegativeprofit}
Assume \eqref{Nmbtalp}. If
\begin{equation}\label{plusprNsuff}
F\leq \frac{\alpha\left((L+L_g)q+g\right)L}{(L+L_g)N}
\end{equation}
hold, then $\Pi\geq 0$.
\end{prop}
\begin{proof}
The profit \eqref{defprofit} in the symmetric equilibrium becomes
\begin{equation}\label{Profit}
\Pi = \left((L+L_g)q+g\right)(p-mw)-Fw
\end{equation}
Applying \eqref{pmwresult} to \eqref{Profit}, and using \eqref{empeq}, we have
\[
\Pi = \frac{\alpha w\left((L+L_g)q+g\right)\frac{L}{N}-Fw(L+L_g)}{L+L_g-\alpha\left((L+L_g)q+g\right)},
\]
the denominator of which is positive from Proposition \ref{prop:LalpLqgplus}. Therefore, the non-negative profit is equivalent to \eqref{plusprNsuff}.
\end{proof}

\subsection{Responses of equilibrium to policy changes}\label{sec:maintheorems}
Theorem \ref{th:empl} shows how the private employment responds to each policy. See Subsection \ref{subsec:thL} for the proof. 
\begin{theo}\label{th:empl}
Suppose that \eqref{Nmbtalp} holds, then
\begin{align}
&\frac{\partial L}{\partial g} > 0, \label{dLdg}\\
&\frac{\partial L}{\partial\tau} < 0, \label{dLdt}\\
&\frac{\partial L}{\partial L_g} > 0.\label{dLdLg}
\end{align}
\end{theo}
The first two results (\eqref{dLdg} and \eqref{dLdt}) are consistent with the standard IS-LM (or AD-AS) model, in which the government purchase (resp. tax increase) improves (resp. deteriorates) the private employment. The third result \eqref{dLdLg} is also straightforward in the sense that the government employment supports income and, therefore, increases the private employment through economic expansion.

Theorem \ref{th:cons} shows responses of the consumption per capita to the policies. See Subsection \ref{subsec:thq} for the proof. 
\begin{theo}\label{th:cons}
Suppose that \eqref{Nmbtalp} holds, then 
\begin{align}
&\frac{\partial q}{\partial g} < 0, \label{dqdg}\\
&\frac{\partial q}{\partial\tau} < 0, \label{dqdt}\\
&\frac{\partial q}{\partial L_g} \geq 0\label{dqdLg}.
\end{align}
\end{theo}
In this model, the government purchase crowds out the private consumption as  \eqref{dqdg} states. This is because, as shown below in Theorem \ref{th:price}, the government purchase leads to a higher price, which must decrease the private consumption under a given nominal wage. The tax increase also reduces disposable income, heasence, \eqref{dqdt} is intuitive. Meanwhile, it is noteworthy that an increase in the government employment, in contrast to the government purchase, leads to an increase in the consumption, and does not cause crowding out as shown in \eqref{dqdLg}. This is because the government employment does not raise the price, as Theorem \ref{th:price} shows below.

Theorem \ref{th:price} shows responses of the price to the policies. See Subsection \ref{subsec:thp} for the proof. 
\begin{theo}\label{th:price}
Suppose that \eqref{Nmbtalp} holds, then 
\begin{align}
&\frac{\partial p}{\partial g} > 0, \label{dpdg}\\
&\frac{\partial p}{\partial\tau} < 0, \label{dpdt}\\
&\frac{\partial p}{\partial L_g} \leq 0. \label{dpdLg}
\end{align}
\end{theo}
Here, \eqref{dpdg} and \eqref{dpdt} are also consistent with the standard AD-AS model, that is, the government purchase raises the price level, and the tax increase lowers it. We should emphasize \eqref{dpdLg} which shows that the government employment does not lead to a higher price, rather it lowers the price in general. The difference stems from whether the additional effective demand is affected by the pricing behavior of the firms or not. The government purchase provides the firms with a certain amount of demand independent of the prices, allowing the firms to set aggressive prices. On the other hand, since demand from the governmentally employed workers is of course affected by the prices, there is pressure to lower the prices due to competition among the firms.

Theorem \ref{th:profit} and \ref{th:profitLg} show how the profit responds to the policies. See Subsection \ref{subsec:thPigt} and \ref{subsec:thprLgLam} for the proof. 
\begin{theo}\label{th:profit}
Suppose that \eqref{Nmbtalp} holds, then 
\begin{align}
&\frac{\partial \Pi}{\partial g} > 0, \label{dPidg}\\
&\frac{\partial \Pi}{\partial\tau} < 0. \label{dPidt}
\end{align}
\end{theo}

\begin{theo}\label{th:profitLg}
Suppose that \eqref{Nmbtalp} holds. Let $\Lambda$ be
\begin{equation}\label{thrLambda}
\Lambda:=(1-Nmq)\left(\frac{Nm}{1-\tau}+2\alpha\right)g-N(mg+F).
\end{equation}
Then,
\begin{equation}\label{dPidLg}
\left\{
\begin{aligned}
&\frac{\partial \Pi}{\partial L_g} < 0,\text{when } L_g<\Lambda,\\
&\frac{\partial \Pi}{\partial L_g} = 0,\text{when } L_g=\Lambda,\\
&\frac{\partial \Pi}{\partial L_g} > 0,\text{when } L_g>\Lambda.
\end{aligned}
\right.
\end{equation}
\end{theo}

The government purchase increases the profit as in \eqref{dPidg} while the tax increase reduces the profit as in \eqref{dPidt}. The impact of the government employment depends on the size of the government employment itself, as in \eqref{dPidLg}. If the thereshold $\Lambda$ takes a negative value, then $\partial\Pi/\partial L_g>0$ for any $L_g$. If $L_g\leq \Lambda$, then the firms may prefer the government purchase that secures their own profits over the government employment. For $L_g>\Lambda$, it is unclear which policy the firms prefer, but the following theorem states that if the tax rate $\tau$ is sufficiently high, the government purchase is more favorable to the firms. See Subsection \ref{subsec:thPigPiLg} for the proof.
\begin{theo}\label{th:profitgorLg}
Suppose that \eqref{Nmbtalp} holds. If
\begin{equation}\label{alpgt1tau}
\alpha >1-\tau
\end{equation}
holds, then
\[
\frac{\partial \Pi}{\partial g} > \frac{\partial \Pi}{\partial L_g}.
\]
\end{theo}

\section{Conclusion and discussion}
We have devised a simple non-Walrasian general equilibrium model and investigated analytically how the macro-economic policies (the government purchase, tax-rate operation, and government employment) affect the private employment, consumption per capita, and profit. While the model shares many similar conclusions with the standard IS-LM or AD-AS model for the effects of the government expenditure and tax on these macroeconomic variables, it clearly shows the difference between the government purchase and the government employment.

Of particular importance is that the government purchase is inflationary in the sense that additional effective demand created by the government purchase not only increases the private employment as shown in \eqref{dLdg} of Theorem \ref{th:empl} but also raises the price as shown in \eqref{dpdg} of Theorem \ref{th:price}. Therefore, relatively large amounts of additional government purchase are required to achieve full employment. Moreover, the inflation caused in this way reduces the consumption per capita as shown in \eqref{dqdg} of Theorem \ref{th:cons}. 

On the other hand, the government employment is {\it not} inflationary as shown in \eqref{dpdLg} of Theorem \ref{th:empl}. The government employment is free from the tradeoff between inflation and unemployment and achieves full employment directly. This has the welfare advantage of not reducing the consumption per capita as in \eqref{dqdLg} of Theorem \ref{th:cons}.

\section{Appendix}\label{sec:appendix}
\subsection{Proof of Theorem \ref{th:empl}}\label{subsec:thL}
\begin{proof}
From \eqref{solvedL}, we see that
\begin{align}
&\frac{\partial L}{\partial g} = \frac{N^2m^2}{\alpha+(Nm-\alpha)\tau}>0,\label{dLdgform}\\
&\frac{\partial L}{\partial\tau} = -\frac{N^2m^2\left(L_g+(Nm-\alpha)g+NF\right)}{\left(\alpha+(Nm-\alpha)\tau\right)^2} < 0,\nonumber
\end{align}
and
\begin{equation}\label{dLdLgN1mt}
\frac{\partial L}{\partial L_g} = \frac{N(1-\tau)m}{\alpha+(Nm-\alpha)\tau} > 0.
\end{equation}
\end{proof}

\subsection{Proof of Theorem \ref{th:cons}}\label{subsec:thq}
\begin{proof}
From \eqref{eqq}, we have
\begin{equation}\label{dgdq1tNmalp1t}
\begin{aligned}
&\frac{\partial q}{\partial g} = 
\frac{1-\tau}{Nm+\alpha(1-\tau)}\frac{\alpha\left(-L-L_g+g\frac{\partial L}{\partial g}\right)}{\left(L+L_g\right)^2}.
\end{aligned}
\end{equation}
From \eqref{solvedL} and \eqref{dLdgform}, we see that
\begin{align}
&-L-L_g+g\frac{\partial L}{\partial g}\\
&= \frac{-N(1-\tau)\left\{mL_g+\alpha F\right\}-N^2mF-L_g(\alpha+(Nm-\alpha)\tau)}{\alpha+(Nm-\alpha)\tau} \\
&<0. \label{mLmLgpgdLdgneg}
\end{align}
It follows from \eqref{mLmLgpgdLdgneg} and \eqref{dgdq1tNmalp1t} that $\frac{\partial q}{\partial g} < 0$.

From \eqref{eqq} and \eqref{dLdt}, we have
\begin{align}
\frac{\partial}{\partial \tau}\ln q
& = \left(\frac{1}{L+L_g-\alpha g}-\frac{1}{L+L_g}\right)\frac{\partial L}{\partial\tau}
\\
&\hspace{5mm} -\frac{Nm}{(1-\tau)(Nm+\alpha(1-\tau))}  < 0 \label{1LLgmalgm1LLgmnm1mtNm}
\end{align}
It follows from \eqref{1LLgmalgm1LLgmnm1mtNm} that $\frac{\partial q}{\partial\tau} < 0$ because $q>0$ and $\frac{\partial}{\partial \tau}\ln q=\frac{1}{q}\frac{\partial q}{\partial \tau}$.

From \eqref{eqq}, we see that
\begin{equation}\label{dqdLg1mtNmal1talgdLdLgp1}
\begin{aligned}
\frac{\partial q}{\partial L_g}
&= \frac{1-\tau}{Nm+\alpha(1-\tau)}\frac{\alpha g\left(\frac{\partial L}{\partial L_g}+1\right)}{\left(L+L_g\right)^2}.
\end{aligned}
\end{equation}
It follows from \eqref{dLdLg} and \eqref{dqdLg1mtNmal1talgdLdLgp1} that $\frac{\partial q}{\partial L_g} \geq 0$ in which the equality holds if and only if $g=0$.

\end{proof}

\subsection{Proof of Theorem \ref{th:price}}\label{subsec:thp}
\begin{proof}
From \eqref{p1}, we see that
\begin{equation}\label{dpdgalpmwalpytNmdLdgqLLgLLgalg2}
\frac{\partial p}{\partial g} 
= \alpha\left[mw+\frac{\alpha(y-t)}{N}\right]
\frac{L+L_g-g\frac{\partial L}{\partial g}}{\left(L+L_g-\alpha g\right)^2}.
\end{equation}
From \eqref{mLmLgpgdLdgneg}, we obtain
\begin{equation}\label{Nmal1tLgFNNm1talpalpNmalptpos}
L+L_g-g\frac{\partial L}{\partial g} > 0. 
\end{equation}
It follows from \eqref{dpdgalpmwalpytNmdLdgqLLgLLgalg2} and \eqref{Nmal1tLgFNNm1talpalpNmalptpos} that $\frac{\partial p}{\partial g} > 0$. 

From \eqref{p2}, we see that
\begin{equation}\label{dpdtprep}
\begin{aligned}
&\frac{\partial p}{\partial\tau} &= 
mw\alpha\frac{-g\frac{\partial L}{\partial \tau}+(L+L_g)^2\frac{\partial q}{\partial\tau}}{\left[\left(L+L_g\right)-\alpha\left(\left(L+L_g\right)q+g\right)\right]^2}.
\end{aligned}
\end{equation}
We see from \eqref{L} that 
\begin{equation}\label{dLdtNmLgmgFLg1Nmq2dqdt}
\begin{aligned}
\frac{\partial L}{\partial\tau}
= \frac{Nm(L_g+(mg+F)N)}{(1-Nmq)^2}\frac{\partial  q}{\partial\tau},
\end{aligned}
\end{equation}
and 
\begin{equation}\label{LLg2NmgFLg1mNmq2}
(L+L_g)^2
= \frac{\left(N(mg+F)+L_g\right)^2}{\left(1-Nmq\right)^2}.
\end{equation}
From \eqref{dLdtNmLgmgFLg1Nmq2dqdt} and \eqref{LLg2NmgFLg1mNmq2}, we have
\begin{equation}\label{migdLdtpLLLg}
\begin{aligned}
-g\frac{\partial L}{\partial \tau}+(L+L_g)^2\frac{\partial q}{\partial\tau}
&= \frac{\left(N(mg+F)+L_g\right)\left(NF+L_g\right)}{(1-Nmq)^2} \frac{\partial q}{\partial \tau}.
\end{aligned}
\end{equation}
It follows from \eqref{dqdt}, \eqref{dpdtprep}, and \eqref{migdLdtpLLLg} that $\frac{\partial p}{\partial \tau} < 0$.

From \eqref{p1}, we see that
\begin{equation}\label{dpdLgmwalpw1tNmalpgdLdLg1LLgalg}
\frac{\partial p}{\partial L_g}
= \left[mw+\frac{\alpha (y-t)}{N}\right]
\frac{-\alpha g\left(\frac{\partial L}{\partial L_g}+1\right)}{\left(L+L_g-\alpha g\right)^2}.
\end{equation}
It follows from \eqref{dLdLg} and \eqref{dpdLgmwalpw1tNmalpgdLdLg1LLgalg} that $\frac{\partial p}{\partial L_g}\leq 0$ in which the equality holds if and only if $g=0$.  
\end{proof}

\subsection{Proof of Theorem \ref{th:profit}}\label{subsec:thPigt}
\begin{proof}
From \eqref{Profit}, we see that
\begin{equation}\label{dPidgdLLgq1dgpmwLLgqgdpdg}
\frac{\partial\Pi}{\partial g}
= \left(\frac{\partial}{\partial g}\left(\left(L+L_g\right)q\right)+1\right)(p-mw)+\left((L+L_g)q+g\right)\frac{\partial p}{\partial g}.
\end{equation}
From \eqref{eqq}, we see that
\begin{equation}\label{LplusLgq}
\left(L+L_g\right)q=\frac{(1-\tau)\left(L+L_g-\alpha g\right)}{Nm+\alpha(1-\tau)}.
\end{equation}
By using \eqref{LplusLgq} with \eqref{dLdgform}, we obtain 
\begin{equation}\label{ddgLLgq}
\begin{aligned}
\frac{\partial}{\partial g}\left(\left(L+L_g\right)q\right)
&= \frac{(1-\tau)(Nm-\alpha)}{\alpha+(Nm-\alpha)\tau} > 0.
\end{aligned}
\end{equation}
It follows from \eqref{p-mw>0}, \eqref{dpdg}, and \eqref{ddgLLgq} that $\frac{\partial \Pi}{\partial g}>0$.

From \eqref{Profit}, we see that
\begin{equation} \label{dPidtpre}
\frac{\partial \Pi}{\partial\tau} 
= \frac{\partial}{\partial\tau}\left(\left(L+L_g\right)q\right)\cdot(p-mw)+\left((L+L_g)q+g\right)\frac{\partial p}{\partial\tau}.
\end{equation}
It is easy to see that
\begin{equation}\label{ddtLLgq}
\frac{\partial}{\partial\tau}\left(\left(L+L_g\right)q\right)
=\frac{\partial L}{\partial\tau}q+(L+L_g)\frac{\partial q}{\partial \tau} < 0
\end{equation}
because of \eqref{dLdt} and \eqref{dqdt}. It follows from \eqref{p-mw>0}, \eqref{dpdt}, \eqref{dPidtpre}, and \eqref{ddtLLgq} that $\frac{\partial \Pi}{\partial \tau}<0$.
\end{proof}

\subsection{Proof of Theorem \ref{th:profitLg}}\label{subsec:thprLgLam}
\begin{proof}
From \eqref{pmwresult}, we can rewrite the profit \eqref{Profit} as
\[
\Pi = \frac{\alpha wm\left\{\left(L+L_g\right)q+g\right\}^2}{L+L_g-\alpha\left(\left(L+L_g\right)q+g\right)} -Fw,
\]
and a careful calculation shows that 
\begin{equation}\label{Pirewritten}
\begin{aligned}
\frac{\partial\Pi}{\partial L_g}
&= \frac{\alpha mw\left\{(L+L_g)q+g\right\}\cl{D}}{\left[L+L_g-\alpha\left(\left(L+L_g\right)q+g\right)\right]^2},
\end{aligned}
\end{equation}
where 
\begin{equation}\label{clDdef}
\begin{aligned}
\cl{D} &:= 2(L+L_g)\frac{\partial}{\partial L_g}\left((L+L_g)q\right)\\
&\hspace{3mm}-\alpha \left(\left(L+L_g\right)q+g\right)\frac{\partial}{\partial L_g}\left((L+L_g)q\right)\\
&\hspace{6mm}-\left(\left(L+L_g\right)q+g\right)\frac{\partial}{\partial L_g}\left(L+L_g\right).
\end{aligned}
\end{equation}
By using \eqref{LplusLgq}, we see that 
\begin{equation}\label{parLLgqLg1tNmalp1tau}
\frac{\partial}{\partial L_g}\left(L+L_g\right)
=
\frac{Nm+\alpha(1-\tau)}{1-\tau}
\frac{\partial}{\partial L_g}\left((L+L_g)q\right).
\end{equation}
Applying \eqref{parLLgqLg1tNmalp1tau} to the third term in \eqref{clDdef}, we obtain
\begin{equation}\label{DEpsilon}
\begin{aligned}
\cl{D}&=
{\cl E} \frac{\partial}{\partial L_g}\left((L+L_g)q\right),
\end{aligned}
\end{equation}
where 
\begin{equation}\label{Epsilon}
{\cl E}:= 
2(L+L_g)
-\left((L+L_g)q+g\right)\left(\alpha + \frac{Nm+\alpha(1-\tau)}{1-\tau}\right)
\end{equation}
By using \eqref{dLdLg} and \eqref{dqdLg} we have 
\[
\begin{aligned}
\frac{\partial}{\partial L_g}\left((L+L_g)q\right) > 0
\end{aligned}
\]
Therefore, the sign of ${\cl D}$ is determined by the sign of ${\cl E}$ from \eqref{DEpsilon}. Applying \eqref{LplusLgq} to \eqref{Epsilon}, we obtain 
\begin{equation}\label{clENmclFNmal1tau}
{\cl E} = \frac{Nm\cl{F}}{Nm+\alpha(1-\tau)},
\end{equation}
where ${\cl F}$ is defined by
\begin{equation}\label{clF}
{\cl F} := L+L_g-g\left(2\alpha+\frac{Nm}{1-\tau}\right)
\end{equation}
It follows from \eqref{clENmclFNmal1tau} that  the sign of ${\cl E}$ is determined by ${\cl F}$. By using \eqref{L} and \eqref{clF}, we see that
\begin{equation}\label{clFLgmLmabda}
{\cl F} = \frac{L_g-\Lambda}{1-Nmq},
\end{equation}
where $\Lambda$ is given by \eqref{thrLambda}. Because of \eqref{1mNmqgt0} and \eqref{clFLgmLmabda}, the sign of $\cl F$, and thus $\frac{\partial\Pi}{\partial L_g}$ is determined by $L_g-\Lambda$. This immediately completes the proof.
\end{proof}

\subsection{Proof of Theoerm \ref{th:profitgorLg}}\label{subsec:thPigPiLg}
\begin{proof}
It is easy from \eqref{Profit} to see that
\begin{equation}\label{dPidgddgLLgqg}
\frac{\partial \Pi}{\partial g}
= \frac{\partial}{\partial g}\left(\left(L+L_g\right)q + g\right)(p-mw)
+ \left(\left(L+L_g\right)q + g\right)\frac{\partial p}{\partial g}
\end{equation}
and 
\begin{equation}\label{dPidLgddLgLLgqgpmw}
\frac{\partial \Pi}{\partial L_g}
= \frac{\partial}{\partial L_g}\left(\left(L+L_g\right)q + g\right)(p-mw)
+ \left(\left(L+L_g\right)q + g\right)\frac{\partial p}{\partial L_g}
\end{equation}
It is obvious from \eqref{dpdg} and \eqref{dpdLg} that 
\begin{equation}\label{dpdg>dpdLg}
\frac{\partial p}{\partial g} > \frac{\partial p}{\partial L_g}.
\end{equation}
From \eqref{dPidgddgLLgqg}, \eqref{dPidLgddLgLLgqgpmw}, and \eqref{dpdg>dpdLg}, if
\begin{equation}\label{ddgLLgqggtddLgLLgqg}
\frac{\partial}{\partial g}\left(\left(L+L_g\right)q + g\right)  > \frac{\partial}{\partial L_g}\left(\left(L+L_g\right)q + g\right),
\end{equation}
then we obtain $\frac{\partial \Pi}{\partial g} > \frac{\partial \Pi}{\partial L_g}$.

By using \eqref{Nmbtalp} and \eqref{ddgLLgq}, we have
\begin{equation}\label{dLLgqgdgNmalNmalt}
\frac{\partial}{\partial g}\left(\left(L+L_g\right)q + g\right)
= \frac{Nm}{\alpha + (Nm-\alpha)\tau} \geq 0,
\end{equation}
and by using \eqref{Nmbtalp}, \eqref{dLdLgN1mt}, and \eqref{parLLgqLg1tNmalp1tau}, we have
\begin{equation}\label{dLgqgdg1mtalNmalt}
\frac{\partial}{\partial L_g}\left(\left(L+L_g\right)q + g\right)
= \frac{1-\tau}{\alpha + (Nm-\alpha)\tau} \geq 0.
\end{equation}
Thus, \eqref{Nmbtalp}, \eqref{alpgt1tau}, \eqref{dLLgqgdgNmalNmalt}, and \eqref{dLgqgdg1mtalNmalt} imply \eqref{ddgLLgqggtddLgLLgqg} which completes the proof.
\end{proof}

\vspace{5mm}
\noindent
{\bf \Large Acknowledgments}\vspace{4mm} \\
\noindent
The author would like to express sincere gratitude to the anonymous reviewers for their valuable comments. In particular, one of the reviewers pointed out an unnecessary assumption in Proposition \ref{prop:nonnegativeprofit} and also provided a comment that led the author to formulate Theorem \ref{th:profitgorLg}.

\bibliographystyle{aer}
\ifx\undefined\bysame
\newcommand{\bysame}{\leavevmode\hbox to\leftmargin{\hrulefill\,\,}}
\fi

\end{document}